\documentclass[conference]{IEEEtran}
\pdfoutput=1
\usepackage{amssymb,amsmath,amsfonts}
\usepackage[pdftex]{graphicx}
\usepackage{cite}
\usepackage{subfigure}

\newtheorem{theorem}{Theorem}

\newtheorem{definition}{Definition}
\newtheorem{example}{Example}

\newcommand{\set}[1]{{\mathcal{#1}}}
\newcommand{\defined}{\triangleq}
\def\ker{{C}}
\DeclareMathOperator{\inc}{in}
\DeclareMathOperator{\head}{head}
\DeclareMathOperator{\tail}{tail}
\DeclareMathOperator{\nbr}{n}

\newcommand{\support}[1]{{\lambda\left(#1\right)}}

\def\field{{\mathbb{F}}}

\def\x{{X}}
\def\y{{Y}}

\newcommand{\nodes}{\set{N}}
\newcommand{\edges}{\set{L}}
\newcommand{\sources}{\set{K}}

\begin{document}
\title{Decoding Network Codes by Message Passing}

\author{\authorblockN{Daniel Salmond\authorrefmark{1},
Alex Grant\authorrefmark{2}, Terence Chan\authorrefmark{2} and Ian
Grivell\authorrefmark{1} }
\authorblockA{\authorrefmark{1}C3I Division, Defence Science \&
  Technology Organisation, Adelaide, South Australia \\Email:
  \{daniel.salmond, ian.grivell\} @dsto.defence.gov.au}
\authorblockA{\authorrefmark{2}Institute for Telecommunications
  Research, University of South Australia, Adelaide, South Australia\\
Email: \{alex.grant, terence.chan\}@unisa.edu.au} }

\date{\today}
\maketitle

\begin{abstract}
  In this paper, we show how to construct a factor graph from a
  network code.  This provides a systematic framework for decoding
  using message passing algorithms.  The proposed message passing
  decoder exploits knowledge of the underlying communications network
  topology to simplify decoding.  For uniquely decodeable linear
  network codes on networks with error-free links, only the message
  supports (rather than the message values themselves) are required to
  be passed. This proposed simplified support message algorithm is an
  instance of the sum-product algorithm. Our message-passing framework
  provides a basis for the design of network codes and control of
  network topology with a view toward quantifiable complexity
  reduction in the sink terminals.
\end{abstract}

\section{Introduction}

Network coding \cite{AhlCai00} is a generalization of routing where
intermediate nodes forward coded combinations of received packets
(rather than simply switching). This approach yields many advantages,
which by now are well documented in the
literature~\cite{YeuLi06,FraSol08,HoLun08}.  For the single
session multicast problem, it is well known that linear network codes
are optimal and can achieve the fundamental limit characterized by the
min-cut bound~\cite{LiYeu03}. Linear network codes are strongly
motivated by practical considerations, namely the implementation of
encoders and decoders. Encoders compute linear combinations of
received packets, and decoding may be achieved by solving a system of
linear equations (since the sink terminals receive a linear
transformation of the source data). The standard approach for decoding
is Gaussian elimination followed by back-substitution. The resulting
decoding complexity is cubic in the size of the linear system, and is
essentially independent of the underlying topology.

However, the topology of the underlying communications network has a
direct impact on the structure of the linear system that requires
solution. This admits the possibility of faster decoding algorithms
which exploit such structure. As one example, it is well known that
band-diagonal systems may be solved with complexity that scales
quadratically in the bandwidth \cite{GolvLo96}. Similarly, there are
many low-complexity iterative solvers for large sparse linear systems
\cite{Axe94} (although it should be noted that such iterative methods
are usually confined to operating over real or complex fields).  This
paper is motivated by the possibility of faster decoding algorithms,
which exploit knowledge of the network topology.

Let a directed acyclic graph, $G = (\nodes, \edges)$ model a
communications network with nodes $\nodes$ and directed links
$\edges$.  We assume that there are $|\sources|$ sources, generating
source symbols $\y_{s}, s\in\sources$ at nodes $a(s)\in\nodes$. Let
$\y_{l}$ be the (network coded) symbol transmitted along link
$l\in\edges$.

For any link $l\in\edges$, define
\begin{multline*}
  \inc(l) \defined \\ \left\{e\in\edges : \head(e)=\tail(l)\right\} \cup
  \left\{s\in\sources: a(s)=\tail(l)\right\}.
\end{multline*}
%the union of the set
%of links $e$ such that $head(e)= tail(l)$ and the set of sources $s$
%such that the source symbol $\y_{s}$ is generated at $tail(l)$.
To simplify notation, we will use the following conventions: (1) for
any undirected graph $\nbr(v)$ is the set of neighboring nodes of $v$,
(2) for any function $g(x)$, we denote its support by
$\support{g}=\{x: g(x) \neq 0 \}$ and (3) for any function $g(x,y)$
and for each $y$, we denote the set $\support{g(x|y)}=\{x: g(x,y) \neq
0\}$.

We are particularly interested in vector linear network codes over a
finite field $\field$, where $y_s,y_l\in\field^n$ for $s\in\sources$
and $l\in\edges$.  Encoding is according to
\begin{equation*}
  y_{l}= \sum_{e\in\inc(l)} c_{l,e}y_{e},\quad l\in\edges
\end{equation*}
where the $c_{l,e}$ are the local encoding coefficients and it is
assumed that $y_{s},s\in\sources$ are independent and uniformly
distributed over $\field^{n}$.

The following example motivates careful exploitation of network
structure by the decoding algorithm.
\begin{example}
  Consider the network in Figure \ref{fig:computing}with sources
  $y_{1},\ldots,y_{K}$.  The receiver $t$ aims to reconstruct all source
  symbols.
    Clearly, decoding can be done by directly solving
  \begin{equation*}
    \begin{bmatrix}
      y_{L-K+1}\\ y_{L-K+2}\\ \vdots\\ y_{L}
    \end{bmatrix}
  = A
    \begin{bmatrix}
      y_{1}\\ y_{2}\\ \vdots\\ y_{K}
    \end{bmatrix}
  \end{equation*}
  where the $K\times K$ matrix $A$ depends on the choice of local
  encoding coefficients. Direct solution involves inversion of the
  $K\times K$ matrix $A$, which is $O(K^{3})$.  On the other hand, due
  to the network topology, it is easy to show that %$A$ can be written
  $A = B_{K-1} \cdots B_{2} B_{1}$ where each matrix $B_{i}$ has at
  most $K+2$ non-zero entries, located on the diagonals and at the
  $(i,i+1)$ and $(i+1,i)$ entries. Inversion of the $B_{i}$ is very
  simple (essentially having the same complexity as inverting a $2
  \times 2 $ matrix).
 Therefore, we can solve the system of linear equations by computing
\begin{align}\label{eq:inverseseq}
\left[
\begin{array}{c}
y_{1} \\ y_{2} \\ \vdots \\ y_{K}
\end{array}
\right]
= B_{1}^{-1} \cdots B_{K-1}^{-1}
\left[
\begin{array}{c}
y_{L-K+1} \\ y_{L-K+2} \\ \vdots \\ y_{L}
\end{array}
\right].
\end{align}
  Thus for the particular network topology of Figure
  \ref{fig:computing}, decoding may actually be achieved in $O(K)$.
\begin{figure}[htbp]
    \begin{center}
      \includegraphics[width=0.9\columnwidth]{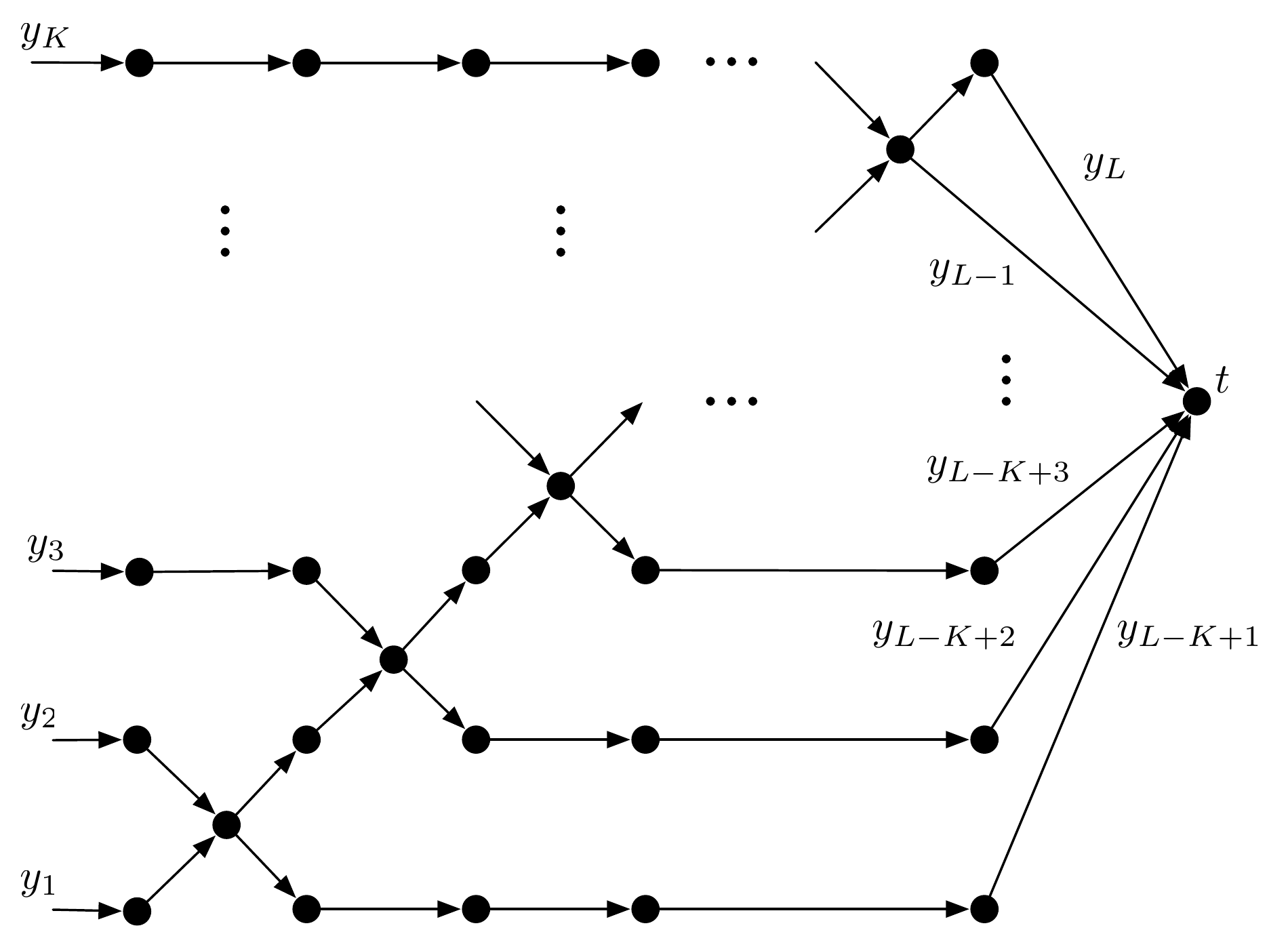}\label{fig:computing}
      \caption{A simple network}
      \label{default}
    \end{center}
  \end{figure}
\end{example}

The obvious fundamental question is how to reduce decoding complexity,
exploiting the topology of an arbitrary network. We propose to use
network topology to construct a factor graph for the network code,
resulting in a simple message-passing decoding algorithm. Despite the
simplicity of this idea, to the best of our knowledge, this paper is
the first systematic development of these methods for linear network
codes. Iterative decoding in a network coding setting has been briefly
considered in~\cite{YanKoe07}, however this work was for a specific
network topology, and did not develop the idea more generally. Iterative decoding has also been considered in~\cite{MoUr07} in which the decoding structure was imposed by a Low Density Parity Check degree distribution instead of the topology of the network.

Section~\ref{sec:factor-graph-message} provides the necessary
background on factor graphs, and introduces a new ``support passing''
algorithm. Roughly speaking, this is the sum-product algorithm for
scenarios where we don't care about the ``probabilities'' themselves,
only whether they are non-zero. In Section~\ref{sec:repr-netw-codes}
we describe how to construct factor graphs from network codes, and
show that the messages utilized by the support-passing algorithm
are cosets of linear subspaces. Section~\ref{sec:junct-tree-algor}
describes how to ensure convergence to the desired solution via
clustering operations.

\section{Factor graphs and message passing}\label{sec:factor-graph-message}

Factor graphs~\cite{KsFrLo01} are a graphical representation of the
factorization of a global function into a product of local
functions. Factor graphs (and related graphical structures such as
Bayesian networks and Tanner graphs) and the attendant message passing
algorithms for marginalization of the global function have found
widespread application~\cite{Fre98}. In particular, these methods are
a key ingredient in modern coding schemes such as low-density parity
check codes~\cite{RiUr08}.

\begin{definition}[Factor graph]
%\comment{I assume we don't want this to be specific to real functions}
  Consider a global function $g$  that factors into a product
  of local functions
  \begin{equation*}
    g(x_{1},\ldots,x_{n}) = \prod_{j\in\set{J}}\phi_{j}(x_{i}:i\in \alpha_{j})
\end{equation*}
where $\alpha_{j}$ is a subset of $\{ x_{1},\ldots,x_{n}\}$.  A factor
graph for $g$ is a bipartite graph $(\set{V}, \set{F}, \set{E})$ with
\emph{variable nodes} $\set{V} \triangleq \{\x_{1},\ldots, \x_{n}\}$,
\emph{factor nodes} $\set{F}\triangleq \{\phi_{j} : j\in\set{J} \}$
and edges
\begin{equation*}
\set{E} \triangleq \{ (\x_{i},\phi_{j}) \in  \set{V} \times \set{F}:  x_{i} \in \alpha_{j} \}.
\end{equation*}
\end{definition}

Direct brute-force computation of the marginals of a global function
is computationally expensive. However if the global function factors in
a simple manner, computing marginals can be simplified using the
well-known sum-product algorithm~\cite{KsFrLo01}, briefly reviewed
below. Recall that the sum and product operations involved are those
of the relevant semi-ring~\cite{AjMc00}.

\begin{definition}[Sum-Product Algorithm]
  At step $k$, messages $\mu^{k}_{\x \to \phi} (x)$ are sent from variable
  nodes $\x$ to a factor nodes $\phi$, and messages $\nu^{k}_{\phi \to
    \x}(x)$ are sent from factor nodes $\phi$ to a variable nodes
  $\x$. The message updating rules are
  \begin{align}
    \mu^{k+1}_{\x \to \phi} (x) &= \prod_{\psi \in \nbr (\x) \setminus \phi }\nu^{k}_{\psi \to \x}(x) \\
    % \nu^{k+1}_{\phi \to \x} (x) & = \sum_{x^{*}: x^{*} \in \nbr
    %   (\phi) \setminus x }\left( \phi (x, x^{*}: x^{*} \in \nbr
    %   (\phi) \setminus x ) \prod_{\x^{*} \in \nbr (\phi) \setminus
    %     \x} \mu^{k}_{\x^{*} \to \phi}(x^{*})\right) \nu^{k+1}_{\phi
    %   \to \x} (x) & = \sum_{ \sim x}\left( \phi (x, x^{*}: \x^{*}
    %   \in \nbr (\phi) \setminus \x ) \prod_{\x^{*} \in \nbr (\phi)
    %     \setminus \x} \mu^{k}_{\x^{*} \to \phi}(x^{*})\right)
    \nu^{k+1}_{\phi \to \x} (x) & = \sum_{ x^{*}: \x^{*} \in \nbr
      (\phi) \setminus \x } h(x, x^{*}: \x^{*} \in \nbr (\phi)
    \setminus \x)
  \end{align}
  where
  \begin{multline*}
    h(x, x^{*}: \x^{*} \in \nbr (\phi) \setminus \x)  = \\
      \phi (x, x^{*}: \x^{*} \in \nbr (\phi) \setminus
      \x ) \prod_{\x^{*} \in \nbr (\phi) \setminus \x} \mu^{k}_{\x^{*}
        \to \phi}(x^{*}).
  \end{multline*}
\end{definition}
A well-known fundamental result (see e.g. \cite{KsFrLo01}) regarding
the optimality of the sum-product algorithm is restated as follows.
\begin{theorem}[Factor tree]\label{thm:mptree}
  If the factor graph is a tree, then after a finite number of
  updates, all messages will remain unchanged.  Furthermore, for any
  $x_{a}$,
\begin{align}\label{eq:cyclefreemp}
\sum_{x_{i} : i\neq a} g(x_{1},\ldots , x_{n}) = \prod_{\phi \in \nbr(\x_{a}) }\nu^{k}_{\phi \to \x_{a}}(x_{a})
\end{align}
\end{theorem}
Theorem \ref{thm:mptree} shows that when the underlying factor graph
is a tree, the sum-product algorithm yields the desired marginals by
\eqref{eq:cyclefreemp}. If the underlying graph is not a tree, then
the sum-product algorithm may not converge, and if it does converge,
it may converge to an incorrect solution.

As we shall see in Section \ref{sec:repr-netw-codes}, there are
meaningful scenarios where we are not interested in the \emph{value} of the
marginals, but rather their \emph{support}. In such
cases, we can simplify the sum-product algorithm.

\begin{theorem}[Support passing algorithm]\label{thm:support}
  Assume all local functions $\phi_{j}(x_{i}:i\in\alpha_{j})$ are
  nonnegative.
  %\comment{What does this mean for finite fields?}
  Let
  $\mu^{k}_{\x \to \phi}$ and $\nu^{k}_{\phi \to \x}$ be messages
  being passed  by the sum-product algorithm. Then
\begin{align}\label{eq:supportupdateone}
\support{\mu^{k+1}_{\x \to \phi}} &= \bigcap_{\psi \in \nbr(\x) \setminus \phi} \support{ \nu^{k}_{\psi \to \x}},\\
\label{eq:supportupdatetwo}
\support{ \nu^{k+1}_{\phi \to \x} } &= \bigcup_{x^*\in \lambda(\mu_{X^*\rightarrow\phi}) \atop \forall x^*:X^*\in n(\phi) \backslash X } \support{ \phi(x  | x^{*} \in \nbr
  (\phi)\setminus\x )}
\end{align}
and
\begin{equation*}
  \support{\prod_{\psi \in \nbr(\x) }\nu^{k}_{\psi \to \x}(x)} = \bigcap_{\psi \in \nbr(\x) }  \support{\nu^{k}_{\psi \to \x}(x)}.
\end{equation*}
\end{theorem}
\begin{proof}
By direct verification.
\end{proof}
According to Theorem \ref{thm:support}, if only the supports of the
marginals are required, then one can simplify the sum-product
algorithm by passing only the supports of the messages. Furthermore,
the message update rules are given by \eqref{eq:supportupdateone} and
\eqref{eq:supportupdatetwo}.
\begin{theorem}[Convergence of the support
  passing algorithm]\label{thm:convergence}
  Let $g$ be the global function and assume the support passing
  algorithm is used.  Initialize $\support{\mu^{1}_{\x \to \phi}} =
  \support{\nu^{1}_{\phi \to \x}} = \set{X}$ for all variable node
  $\x$, where $\set{X}$ is the alphabet of variable $\x$. Then
  \begin{enumerate}
  \item $\support{\mu^{k+1}_{\x \to \phi}} \subseteq
    \support{\mu^{k}_{\x \to \phi}}$ and $\support{\nu^{k+1}_{\phi \to
        \x}} \subseteq \support{\nu^{k}_{\phi \to \x}}$ for all
    $k$. Hence, after a finite number of iterations, supports of
    message will remain unchanged;

\item
$\support{\sum_{x^{*} : x^{*}\neq x} g(x_{1},\ldots , x_{n})}$ is a subset of
$ \support{\mu^{k}_{\x \to \phi}}$ and $\support{\nu^{k}_{\phi \to \x}}$.
\end{enumerate}
\end{theorem}
\begin{proof}
A direct consequence of  \eqref{eq:supportupdateone} and \eqref{eq:supportupdatetwo}.
\end{proof}

Theorem \ref{thm:convergence} guarantees convergence of the support
passing algorithm, even for graphs with cycles. If the factor graph is
a tree, the algorithm converges to the support of the
marginals. However, if the factor graph contains cycles, it still
converges, but to a limit that contains the support of the
marginals. Thus the presence of cycles can cause convergence to an
undesired solution (although it will ``contain'' the desired solution).

%%%%%%%%%%%%%%%%%%%%%%%%%%%%%%%%%%%%%
\section{Factor Graphs for Network Codes}\label{sec:repr-netw-codes}

Let $Y_{s}$ be the symbol generated by the source $s\in\sources$ and
$Y_{l}$ be the symbol transmitted on link $l\in\edges$.  For each link
$l$, the node $\tail(l)$ receives symbols generated from the sources
or transmitted from its neighboring nodes. Considering stochastic
encoding, the symbol transmitted on link $l\in\edges$, namely $y_{l}$,
is randomly selected according to a conditional distribution
$\ker_{l}(y_{l} | y_{a}, a \in \inc(l))$. Deterministic codes are
obtained via degenerate $C_l$. Note that this set-up also permits
modeling of noisy links (incorporating random errors into the $C_l$).

The probability distribution of  $(Y_{s},Y_{l}, s\in\sources,l \in\set{L})$ is
\begin{align}
 \Pr(y_{a}:a\in\set{L}\cup\sources ) \propto \prod_{l\in\set{L}} \ker_{l}(y_{l} | y_{\inc(l)}).
\end{align}

Consider a receiver which observes incoming symbols $Y_{j}$ as $y_{j}^{*}$ for $j\in\set{J}$.
Then it is straightforward to prove that %the posterior probability distribution is proportional to the following
\begin{align}
& \Pr(y_{a}:a\in\set{L}\cup\sources | y_{i}^{*}, i\in \set{J}) \nonumber \\
&  \qquad  \propto \left(\prod_{l\in\set{L}} \ker_{l}(y_{l} | y_{\inc(l)}) \right) \left( \prod_{j\in\set{J}} \delta(y_{j} - y_{j}^{*})\right).\label{eq:main}
\end{align}
With error-free links and a uniquely decodeable code (linear or
non-linear), it can easily be shown that decoding is equivalent to
finding the support of the $\{\y_s,s\in\sources\}$ marginal of
\eqref{eq:main}. In more general settings, it may be desired to
compute the marginal posterior probabilities of the source variables.

In the following, we will provide a mechanical approach to represent
the global function \eqref{eq:main} with a factor graph. The method
works for any network codes.  However, as we shall show in the next example,
the obtained factor graph may not be the simplest possible.

\begin{definition}[Network Code Factor Graph]\label{def:ncfg}
  Given the network code defined by stochastic local encoding
  functions $C_l$, define a factor graph with variable nodes $\set{V}
  \triangleq \{ \y_{s}: s\in\sources , \y_{l}:l\in\set{L} \}$ and
  factor nodes $\set{F} \triangleq \{ \phi_{l} : l\in\set{L} \} \cup
  \{ \psi_{j} : j\in\set{J} \}$.  Each factor node $\phi_{l}$ is
  associated with the local function $\ker_{l}(y_{l} | y_{\inc(l)}) $,
  and is connected to variable nodes $\y_{l}$ and $\y_{i}$ where
  $i\in\inc(l)$. Factor node $\psi_{j}$ is associated with the local
  function $\delta(y_{j} - y_{j}^{*})$, and is connected only to the
  variable node $\y_{j}$.
\end{definition}

The topology of the above factor graph depends only on the network
topology but not the particular codes being used. As a result, the
obtained factor graph may not be in its simplest possible form in
general. For example, if some local functions $\phi_{l}$ can be
further factorized into a simpler form, one may be able to further
simplify the factor graph. Furthermore, in the context of network
coding, only the variables generated by sources are of interest. In
that case, it may be possible to prune the factor graph without
affecting the decoding process.

\begin{example}
  Consider the network shown in Figure \ref{fig:egtwo}. The symbols
  generated by the source are $y_{1}$ and $y_{2}$. The receiver $t$
  can observe symbols $y_{4}$ and $y_{5}$.  The factor graph for this
  code (with respect to the given receiver) is depicted in Figure
  \ref{fig:egtwofactor}.

  Depending on the specific choices of encoding functions, this
  factor graph may be simplified. Suppose $y_{3}$ is a function of
  $y_{1}$. Then one can simplify the factor graph by removing the
  link between the factor node $\phi_{3}$ and the variable node
  $\y_{2}$.
  \begin{figure}[htbp]
    \centering
    \subfigure[Network]{
      \includegraphics[scale=0.5]{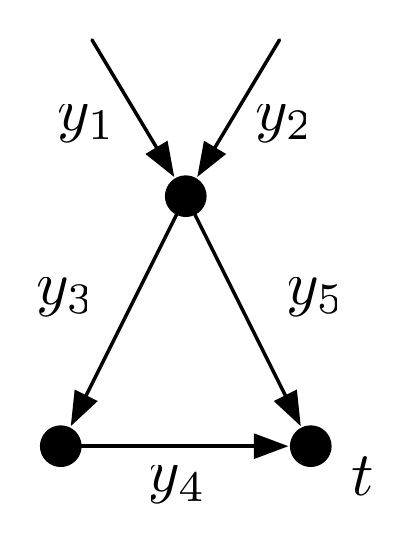}\label{fig:egtwo}
    }
    \subfigure[Factor graph]{
      \includegraphics[scale=0.5]{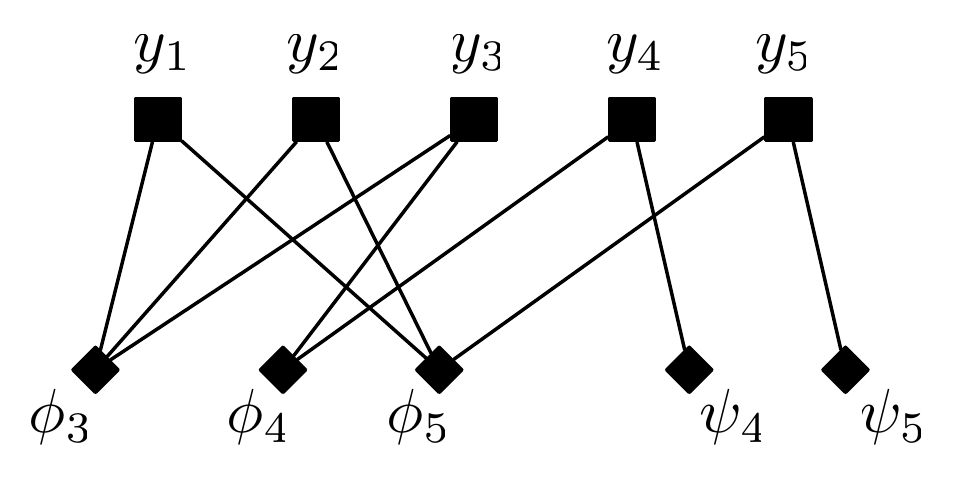}\label{fig:egtwofactor}
    }
    \caption{A relay network}.
  \end{figure}
\end{example}

Let us now consider the special, but important case of deterministic
linear network codes.  All $y_{s}$ for $s\in\sources$ are independent
and uniformly distributed over $\field^{n}$ and for each link
$l\in\set{L}$,
\[
y_{l}= \sum_{e\in\inc(l)} c_{l,e}y_{e}
\]
where the local encoding kernel coefficients $c_{l,e}$ are fixed and known. Hence,\begin{align}
\ker_{l}(y_{l} | y_{\inc(l)}) & =
\begin{cases} 1 & \text{ if }  y_{l}= \sum_{e\in\inc(l)} c_{l,e}y_{e} \\
           0 & \text{ otherwise}
\end{cases} \\
& = \delta \left( y_{l} - \sum_{e\in\inc(l)} c_{l,e}y_{e}\right).
\end{align}

\begin{theorem}[Deterministic linear network codes]\label{thm:dlnc}
  Assume that the support passing algorithm is used.  Suppose that
  $\y_{i}=y_{i}^{*}$ is the actual symbol being transmitted or
  generated. Then
  \begin{enumerate}
  \item $\mu^{k}_{\y \to \phi}$ and $\nu^{k+1}_{\phi \to \y}$ are
    constant over their supports. Therefore, the support passing
    algorithm and the sum-product algorithm are equivalent in this
    case;

  \item $\support{\mu^{k+1}_{\y \to \phi}} $ is a coset of the form
    $y^{*} + W^{k+1}_{\y \to\phi}$ where $W^{k+1}_{\y \to\phi}$ is a
    vector subspace;

  \item Similarly, $\support{\nu^{k+1}_{\phi \to \y}}$ is also a coset
    of the form $y^{*} + \hat{W}^{k+1}_{ \y \to\phi}$ where
    $\hat{W}^{k+1}_{\y \to\phi}$ is also a vector subspace.
\end{enumerate}
Furthermore,
\begin{align}\label{eq:subupdateone}
W^{k+1}_{\y \to\phi} = \bigcap_{\psi: \psi\in\nbr(\y)\backslash \phi} \hat{W}^{k}_{\psi\to\y},
\end{align}
and
\begin{align}\label{eq:subupdatetwo}
\hat{W}^{k+1}_{\phi\to\y} = \left\langle W^{k}_{\y^{*} \to\phi} : \y^{*} \in  \nbr (\phi) \setminus \y  \right\rangle
\end{align}
where
$\left\langle W^{k}_{\y^{*} \to\phi} : \y^{*} \in  \nbr (\phi) \setminus \y  \right\rangle$ is defined as the minimal subspace containing $W^{k}_{\y^{*} \to\phi} $ for all $\y^{*} \in  \nbr (\phi) \setminus \y$.
\end{theorem}

Theorem \ref{thm:dlnc}, shows that when deterministic linear codes are
used, the support passing algorithm can be implemented by ``updating''
subspaces according to \eqref{eq:subupdateone} and
\eqref{eq:subupdatetwo}. In the case that the underlying factor graph is acyclic and the network code is  uniquely decodeable (i.e. invertible global linear transform), this subspace update
algorithm is guaranteed to converge to the correct solution.

\begin{example}
  Consider a butterfly network in Figure \ref{fig:butterflya}. The
  sources are $y_{1}$ and $y_{2}$, and the sinks are denoted by open
  circles. A factor graph for the butterfly network is given in Figure
  \ref{fig:butterflya}. Note that, although the butterfly network
  itself has an undirected cycle, the corresponding factor graph,
  depicted in Figure \ref{fig:butterflyb}, is cycle-free. Therefore,
  message passing and support-passing algorithms are exact.

  Suppose sink $t_{1}$ observes $y_{5}$ and $y_{8}$ and aims to
  reconstruct $y_{1}$ and $y_{2}$. Then some of the links in the
  factor graph Figure \ref{fig:butterflyb} can be removed without
  affecting decoding at $t_1$. This results in the simplified factor
  graph Figure \ref{fig:butterflyc}.
\end{example}

\begin{figure}[ht]
\centering
\subfigure[Butterfly network]{
\includegraphics[scale=0.4]{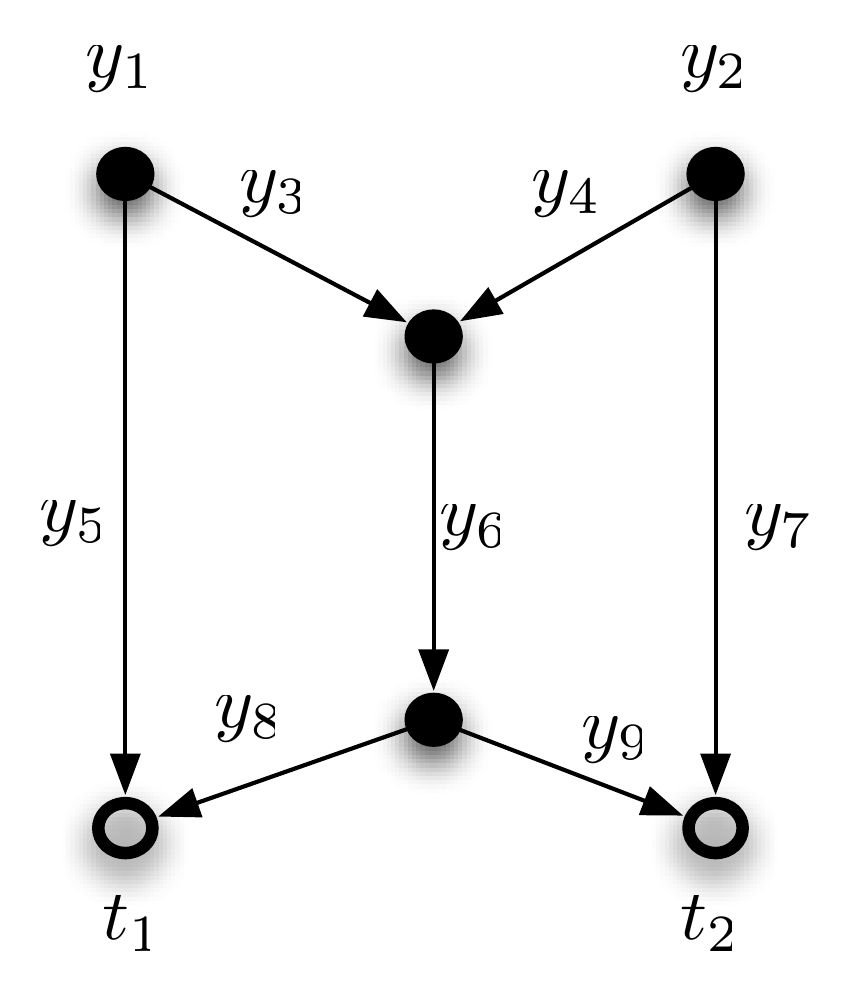}\label{fig:butterflya}
}
\subfigure[Factor graph]{
\includegraphics[scale=0.4]{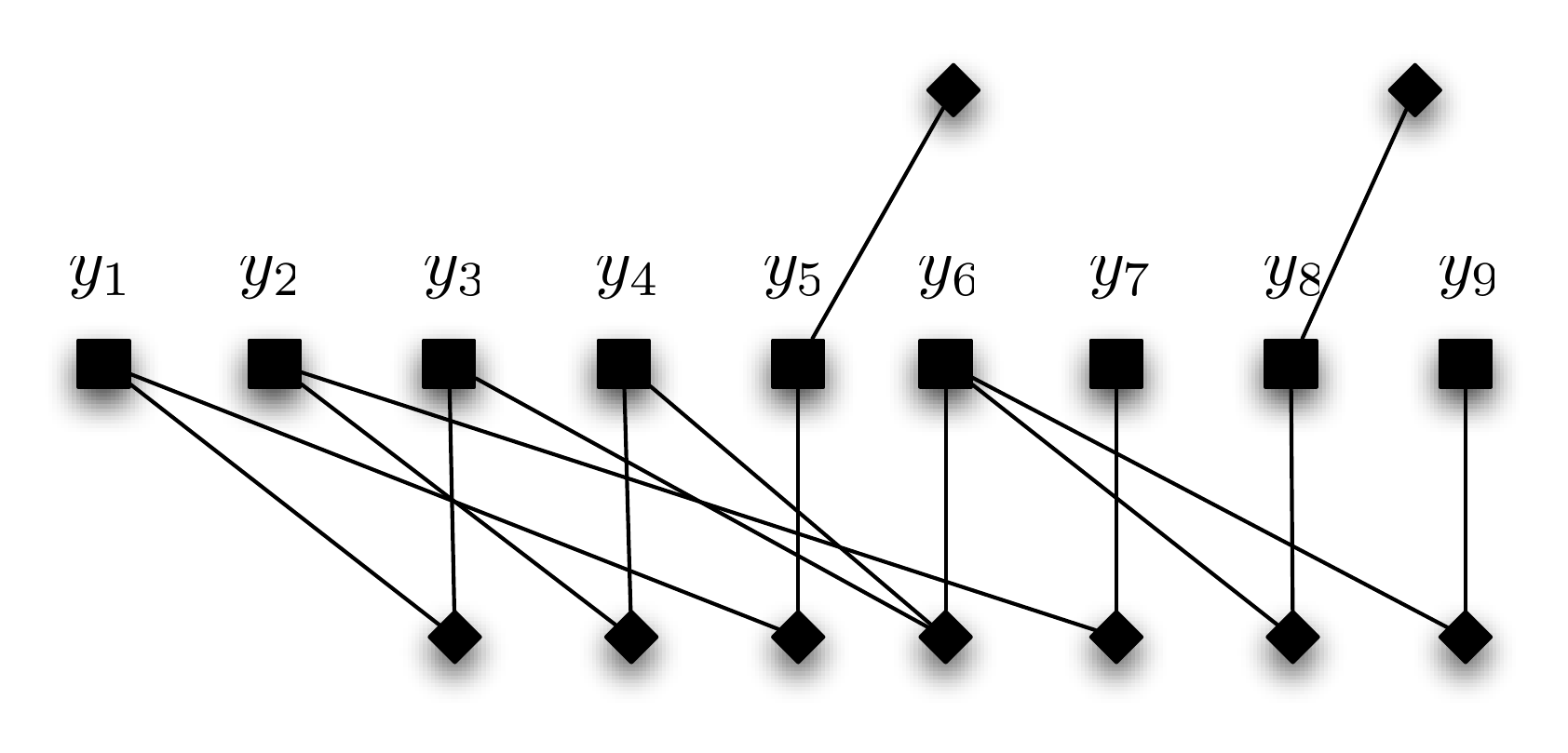}\label{fig:butterflyb}
}
\subfigure[Factor graph (after pruning)]{
\includegraphics[scale=0.4]{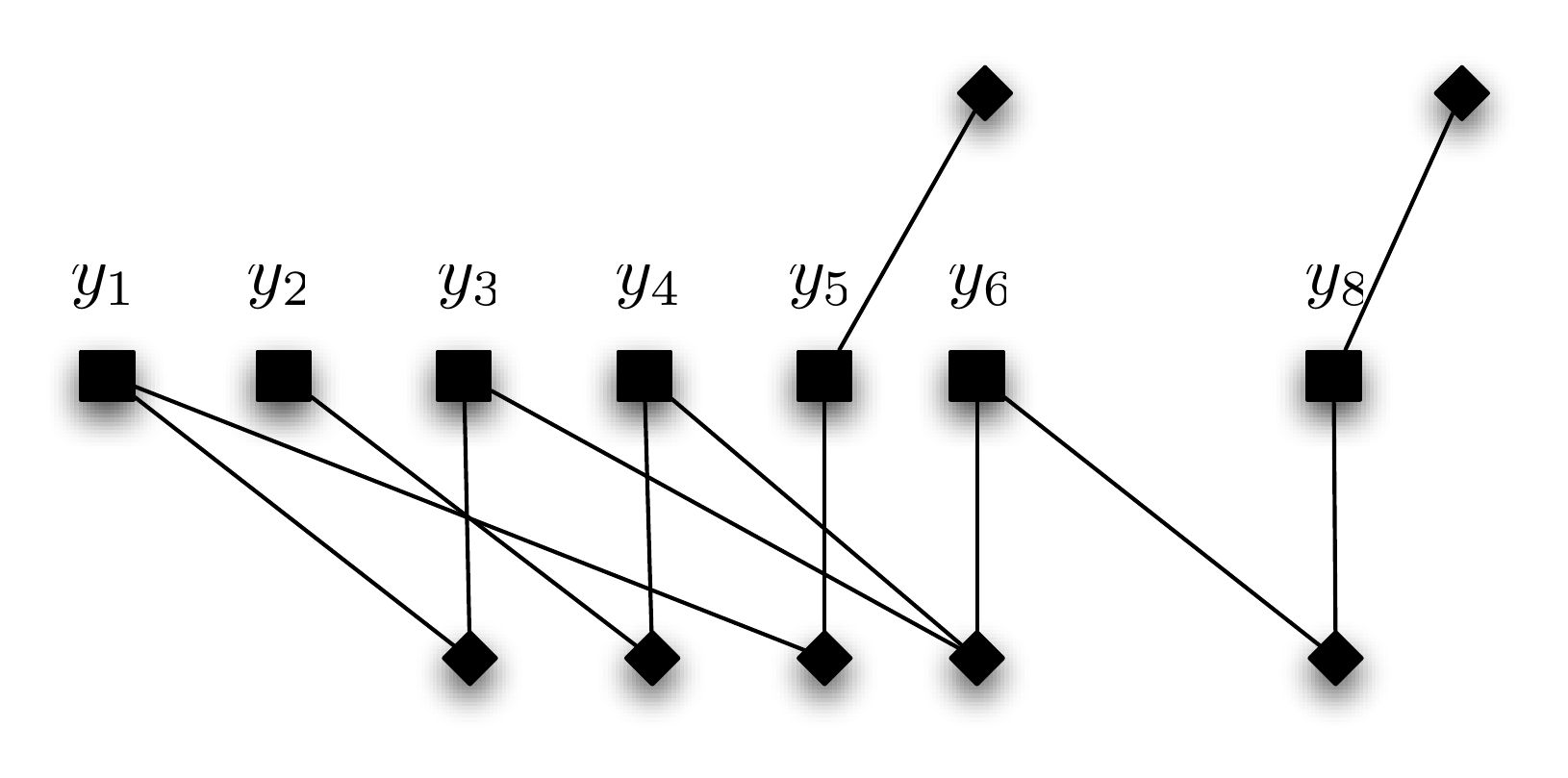}\label{fig:butterflyc}
}
\caption{Butterfly network and its factor graphs}
\end{figure}

\section{Dealing with Cycles}\label{sec:junct-tree-algor}

In the previous section, we proposed a simplified sum-product
algorithm which passes only the supports of messages. It was proved
that the algorithm always converges, and that is equivalent to the
sum-product algorithm in the case of deterministic linear network
codes.  As a result, when the underlying factor graph is a tree, the
proposed support-passing algorithm ensures that the support of
marginals can be found.

Unfortunately, when the factor graph has cycles, the sum-product
algorithm does not always converge. Although support-passing algorithm
converges, it may not converge to the supports of the desired
marginals.

To avoid cycles, one may transform a factor graph with cycles into one
with no cycle.  For example, one common transformation technique is
clustering~\cite{KsFrLo01} as demonstrated in the following
example.

\begin{example}
  Consider the network depicted in Figure \ref{fig:clusterA}, which is
  a simplified version of the network shown in Figure
  \ref{fig:computing}.  From the network, we can construct a factor
  graph in Figure \ref{fig:clusterB} according to Definition
  \ref{def:ncfg}. By clustering function nodes together, we can
  transform the factor graph into a cycle free one as shown in Figure
  \ref{fig:clusterC}. In fact, running the sum-product algorithm over
the cycle free factor graph is essentially equivalent to decoding by
\eqref{eq:inverseseq}.
\end{example}

\begin{figure}[ht]
\centering \subfigure[Original network]{
\includegraphics[scale=0.5]{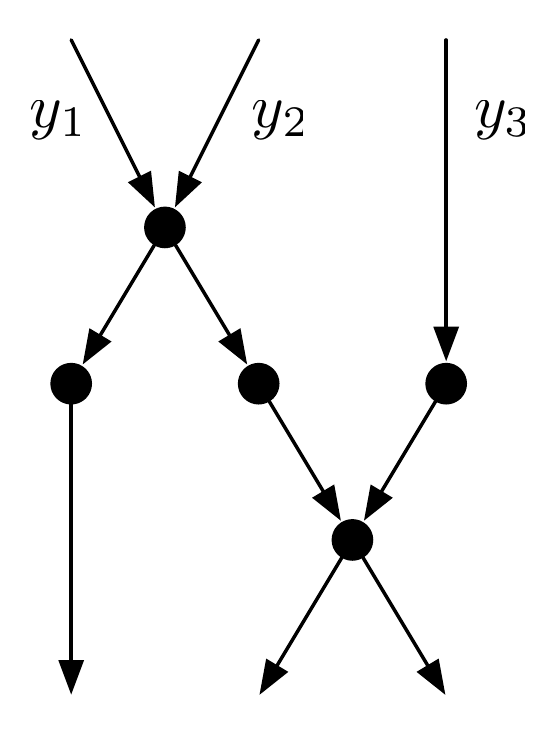}
\label{fig:clusterA}
}

\subfigure[Cyclic factor graph]{
\includegraphics[scale=.5]{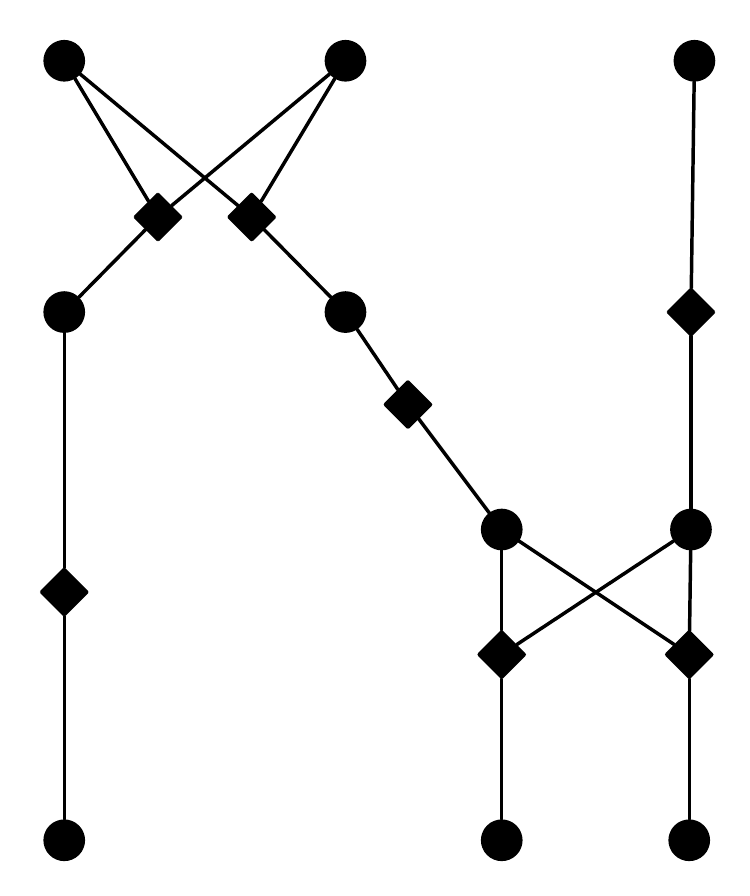}
\label{fig:clusterB}
}
\subfigure[Factor graph after clustering]{
\includegraphics[scale=.5]{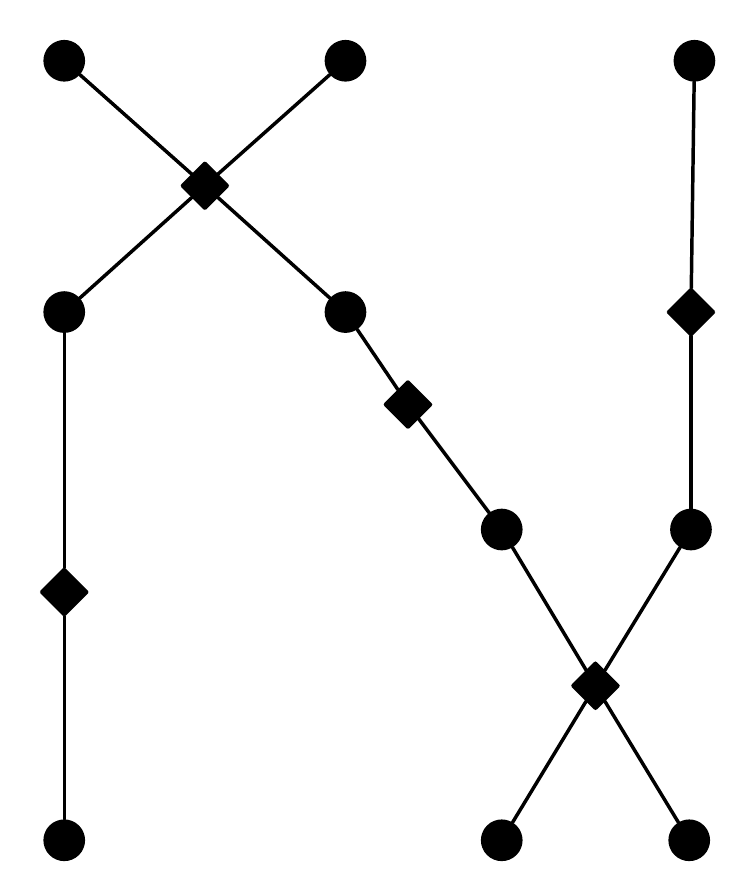}
\label{fig:clusterC} }
\caption{Factor graph transformation}
\end{figure}

Recall that every link is associated with a factor node. As a general
rule, the first step of clustering is to cluster those factor nodes
associated with links originating from the same node.

\section{Conclusion}\label{sec:conclusion}

Decoding of network codes is traditionally achieved by solving a
system of linear equations. Using this approach, the decoding
complexity of the network is essentially independent of the underlying
topology.  This paper shows that if we exploit our knowledge of the
topology, a more efficient decoding algorithm may be obtained.  In
some extreme examples, we showed that the reduction in decoding
complexity can be huge.  This paper prompts a new direction in network
code design: how to choose a network subgraph such that the resulting
network code admits an efficient decoding algorithm.

As the first step towards the goal, we propose the use of message
passing algorithm as a decoding strategy. For a given network code, we
give algorithms to construct a factor graph on which message passing
algorithms (such as the sum-product algorithm) are performed. We
proved that when network links are noiseless, the support passing
algorithm (a simplified version of the sum-product algorithm)
suffices. We showed that the support-passing algorithm is exact when
the underlying factor graph is acyclic and always converges to a limit
which contains the desired solutions, even when the factor graph is
cyclic. Finally, we discussed the use of some graph augmentation
techniques to transform a cyclic factor graph to an acyclic one so
that the support passing algorithm is exact.

\bibliographystyle{ieeetr}
\bibliography{network}

\end{document}